\documentclass[11pt, a4paper]{amsart}

\usepackage{amsmath, amsthm, amsfonts, amssymb}
\usepackage{graphicx}
\usepackage{nicefrac}

\newtheorem{thm}{Theorem}[section]
\newtheorem{lem}[thm]{Lemma}
\newtheorem{cor}[thm]{Corollary}
\theoremstyle{definition}
\newtheorem{defn}[thm]{Definition}
\theoremstyle{remark}
\newtheorem{remark}[thm]{Remark}
\numberwithin{equation}{section}

\DeclareMathOperator{\asinh}{arcsinh}
\newcommand{\Lim}[1]{\raisebox{0.5ex}{\scalebox{0.8}{$\displaystyle \lim_{#1}\;$}}}

\newcommand{\doslineas}[2]{\genfrac{}{}{0pt}{}{#1}{#2}}

\begin{document}

\title[Approximate solutions of the hyperbolic Kepler's equation]{A new method for obtaining approximate solutions of the hyperbolic Kepler's equation}

\author[M.~Avendano]{Mart\'\i n Avendano}
\address{Centro Universitario de la Defensa\\
         Academia General Militar\\
         Ctra. de Huesca s/n\\
         50090, Zaragoza, Spain and IUMA, Universidad de Zaragoza, Spain}
\email{avendano@unizar.es}

\author[V.~Mart\'\i n-Molina]{Ver\'onica Mart\'\i n-Molina}
\address{Departamento de Matem\'aticas and IUMA\\
         Facultad de Educaci\'on. Universidad de Zaragoza, Pedro Cerbuna 12\\
         50009, Zaragoza, Spain}
\email{vmartin@unizar.es}

\author[J.~Ortigas-Galindo]{Jorge Ortigas-Galindo}
\address{Instituto de Educaci\'on Secundaria \'Elaios\\
        Andador Pilar Cuartero 3\\
        50018, Zaragoza, Spain and IUMA, Universidad de Zaragoza, Spain}
\email{jortigas@unizar.es, jortigas@educa.aragon.es}


\begin{abstract}
We provide an approximate zero $\widetilde{S}(g,L)$ for the hyperbolic Kepler's equation $S-g\,\asinh(S)-L=0$ for
$g\in(0,1)$ and $L\in[0,\infty)$. We prove, by using Smale's $\alpha$-theory, that Newton's method starting at our
approximate zero produces a sequence that converges to the actual solution $S(g,L)$ at quadratic speed, i.e.~if
$S_n$ is the value obtained after $n$ iterations, then $|S_n-S|\leq 0.5^{2^n-1}|\widetilde{S}-S|$. The
approximate zero $\widetilde{S}(g,L)$ is a piecewise-defined function involving several linear expressions and one
with cubic and square roots. In bounded regions of $(0,1) \times [0,\infty)$ that exclude a small neighborhood of $g=1, L=0$, we also provide a method to construct simpler starters involving only constants.
\end{abstract}

\maketitle

\section{Introduction}\label{sec-intro}

Keplerian orbits can be described by two parameters, the semilatus rectum $p>0$ and the eccentricity $e\geq 0$, from
which all other orbital constants can be derived. According to Kepler's laws, the coordinates $(x,y)$ at time $t$ of a
satellite in the standard reference frame (central body at the origin and $x$-axis pointing to the periapsis) are given by
\begin{center}
{
\def\arraystretch{1.5}
\begin{tabular}{l|l|l}
Elliptic ($e<1$)                    & Parabolic ($e=1$)            & Hyperbolic ($e>1$)  \\
\hline
$x=\frac{p}{1-e^2}(\cos(E)-e)$       & $x=\frac{p}2(1-D^2)$         & $x=\frac{p}{1-e^2}(\cosh(H)-e)$ \\
$y=\frac{p}{\sqrt{1-e^2}}\sin(E)$    & $y=pD$                       & $y=\frac{p}{\sqrt{e^2-1}}\sinh(H)$ \\
$E-e\sin(E)=M$                       & $D+\frac{D^3}3=M$            & $e\sinh(H)-H=M$ \\
$M=\sqrt{\frac{\mu(1-e^2)^3}{p^3}}(t-t_0)$ & $M=\sqrt{\frac{4\mu}{p^3}}(t-t_0)$ & $M=\sqrt{\frac{\mu(e^2-1)^3}{p^3}}(t-t_0)$
\end{tabular}
}
\end{center}
where $\mu$ is the gravitational parameter of the central body and $t_0$ is the time of passage through periapsis \cite{MP94}. The value $M$, called mean anomaly, is a linear function of the time. $E$, $D$ and $H$ are the elliptic, parabolic and hyperbolic eccentric anomalies respectively, which are functions of the time which are related to $M$ through Kepler's equation (third line of formulas in the table above).
The eccentric anomalies are needed to compute the coordinates $x$ and $y$, hence the need for a method to solve Kepler's equation.

In the case of circular orbits ($e=0$), the equation for the eccentric anomaly reduces to $E=M$, and in the case of
parabolic trajectories ($e=1$), the cubic equation for $D$ can be solved exactly:
\[
  D=\sqrt[3]{\frac{3M+\sqrt{9M^2+4}}{2}}+\sqrt[3]{\frac{3M-\sqrt{9M^2+4}}{2}}.
\]
Therefore, the real problem resides on the cases $0<e<1$ and $e>1$, where the equation involves a combination of algebraic and
trascendental functions.

For simplicity, the hyperbolic Kepler's equation is usually written as $\sinh(H)-gH=L$, where $g=\frac{1}e\in(0,1)$ and
$L=\frac{M}e\in(-\infty,\infty)$. Since the left side of the equation is an odd function of $H$, it is enough to consider
$L\in[0,\infty)$. Moreover, assuming that $H$ can be found, the formula for the $y$ coordinate requires $\sinh(H)$, so it
makes more sense to directly find $S=\sinh(H)$ solving
\begin{equation}\label{eq-hkepler}
f_{g,L}(S)=  S-g\asinh(S)-L=0.
\end{equation}
Once that is done, we can compute $y$ as $y=\frac{p}{\sqrt{e^2-1}}S$ and, since $\cosh(H)=\sqrt{1+\sinh^2(H)}$, $x$ as $x=\frac{p}{1-e^2}(\sqrt{1+S^2}-e)$.

While there are many articles discussing solutions for the elliptic case (see \cite{AMO14}, \cite{Colw}, \cite{DB83}, \cite{MP94}, \cite{Ng}, \cite{OG86}, \cite{TB89}, among others), the hyperbolic case has received less attention.
Prussing~\cite{P77} and Serafin~\cite{S86}  gave upper bounds for the actual solution of the hyperbolic Kepler's equation, which
can be used as starters for Newton's method since $f_{g,L}''\geq 0$.
Gooding and Odell \cite{OG88} solved the hyperbolic Kepler's equation by using Newton's method starting from a well-tuned formula depending on the parameters $g$ and $L$. Their approach gives a relative accuracy of $10^{-20}$ with only two iterations. Although their starter is a single formula that works on the entire region $(0,1)\times[0,\infty)$, it is too complicated to provide an efficient
implementation.

In contrast, Fukushima \cite{F97} focused on the efficiency and simplicity of the starter rather than trying to find a universal formula, and produced a starter that is defined by different formulas, each valid in a stripe-like region. He showed that his starter converges under Newton's method, but not necessarily at quadratic speed.

Our approach to solve Eq.~\eqref{eq-hkepler} is to use Newton's method starting from a value $\widetilde{S}(g,L)$
that is close enough to the actual solution $S(g,L)$ to guarantee quadratic convergence speed, i.e.~if $S_n$ denotes the
value obtained after $n$ iterations, then $|S_n-S|\leq 0.5^{2^n-1}|\widetilde{S}-S|$. We use a simple criterion, Smale's $\alpha$-test~\cite{BCSS} with the constant $\alpha_0$ improved by Wang and Han~\cite{WH89}, to decide whether a starter gives the
claimed convergence rate. Values that satisfy the test are called approximate zeros.

\begin{defn}[Smale's $\alpha$-test] \label{def}
Let $f:(a,b)\subseteq\mathbb{R}\to\mathbb{R}$ be an infinitely differentiable function. The value $z\in(a,b)$ is an
approximate zero of $f$ if it satisfies the following condition
\[
  \alpha(f,z)=\beta(f,z) \cdot \gamma(f,z) <\alpha_0,
\]
where
\[
  \beta(f,z)=\left|\frac{f(z)}{f'(z)} \right|,\quad
  \gamma(f,z)=\sup_{k \geq 2} \left| \frac{f^{(k)}(z)}{k!f'(z)}\right|^{\frac{1}{k-1}}
\]
and $\alpha_0=3-2\sqrt{2}\approx 0.1715728$.
\end{defn}

In an earlier paper \cite{AMO14}, we obtained the following simple approximate zero $\widetilde{E}(e,M)$ for the elliptic case.
\begin{thm}[\cite{AMO14}]\label{thm-starterKepler}
The starter
\[
  \widetilde{E}(e,M)=\left\{
  \begin{array}{cl}
    M & \text{if } e \leq \nicefrac12\text{ or } M \geq \nicefrac{2\pi}{3} \\
    \nicefrac{2\pi}{3} & \text{if } e > \nicefrac12\text{ and } \nicefrac{\pi}{4}\leq M <  \nicefrac{2\pi}{3} \\
    \nicefrac{\pi}{2} & \text{if } e > \nicefrac12\text{ and } \nicefrac{\pi}{7}\leq M < \nicefrac{\pi}{4} \\
    \frac{M}{1-e} & \text{if } e > \nicefrac12,\; M<\nicefrac{\pi}{7} \text{ and } M <\frac{ \sqrt[4]{12\alpha}(1-e)^{\nicefrac32}}{\sqrt{e}} \\
    \frac{\sqrt[3]{6Me^2}}{e}-\frac{2(1-e)}{\sqrt[3]{6Me^2}} & \text{otherwise}
  \end{array}
  \right.
\]
is an approximate zero of $E-e\sin(E)-M$ for all $e\in[0,1)$ and $M\in[0,\pi]$.
\end{thm}

Our main result in this paper, proven in Section \ref{sec-starter}, is a starter for the hyperbolic Kepler's equation that is a piecewise-defined function in eight stripe-like regions. In seven of them, we use only linear expressions in $g$ and $L$, while in the last one we need a cubic and a square root.

\begin{thm}\label{thm-starterKeplerS}
The starter
\[
  \widetilde{S}(g,L)=\left\{
  \begin{array}{ll}
    L+2.30  \, g                                                               & \text{if }\, 4-1.9 \, g       <       L   \\
    L+1.90  \, g                                                               & \text{if }\, 2.74 - 1.56\,g < L  \leq 4 -     1.9\,g    \\
    L+1.56 \, g                                                               & \text{if }\, 2.01 - 1.33\,g < L  \leq 2.74 - 1.56\,g  \\
    L+1.33 \, g                                                               & \text{if }\, 1.60 - 1.16\,g < L  \leq 2.01 - 1.33\,g \\
    L+1.16 \, g                                                               & \text{if }\, 1.32 - 1.02\,g < L \leq  1.60 - 1.16\,g\\
    L+1.02 \, g                                                               & \text{if }\, 1.12 - 0.91\,g < L  \leq 1.32 - 1.02\,g \\
    L+0.91 \, g                                                               & \text{if }\, 1-\frac56  \,g < L  \leq 1.12 - 0.91\,g \\
    \text{the exact solution of } (1-g) \widetilde{S}+g\frac{\widetilde{S}^3}{6}=L    & \text{if }\,  0 \leq L \leq 1-\frac56 \, g\\        \end{array}
  \right.
\]
is an approximate zero of $f_{g,L}$ for all $g\in (0,1)$ and $L\in[0,\infty)$.
\end{thm}
\begin{figure}[!ht]
\centering
\includegraphics[height=5cm]{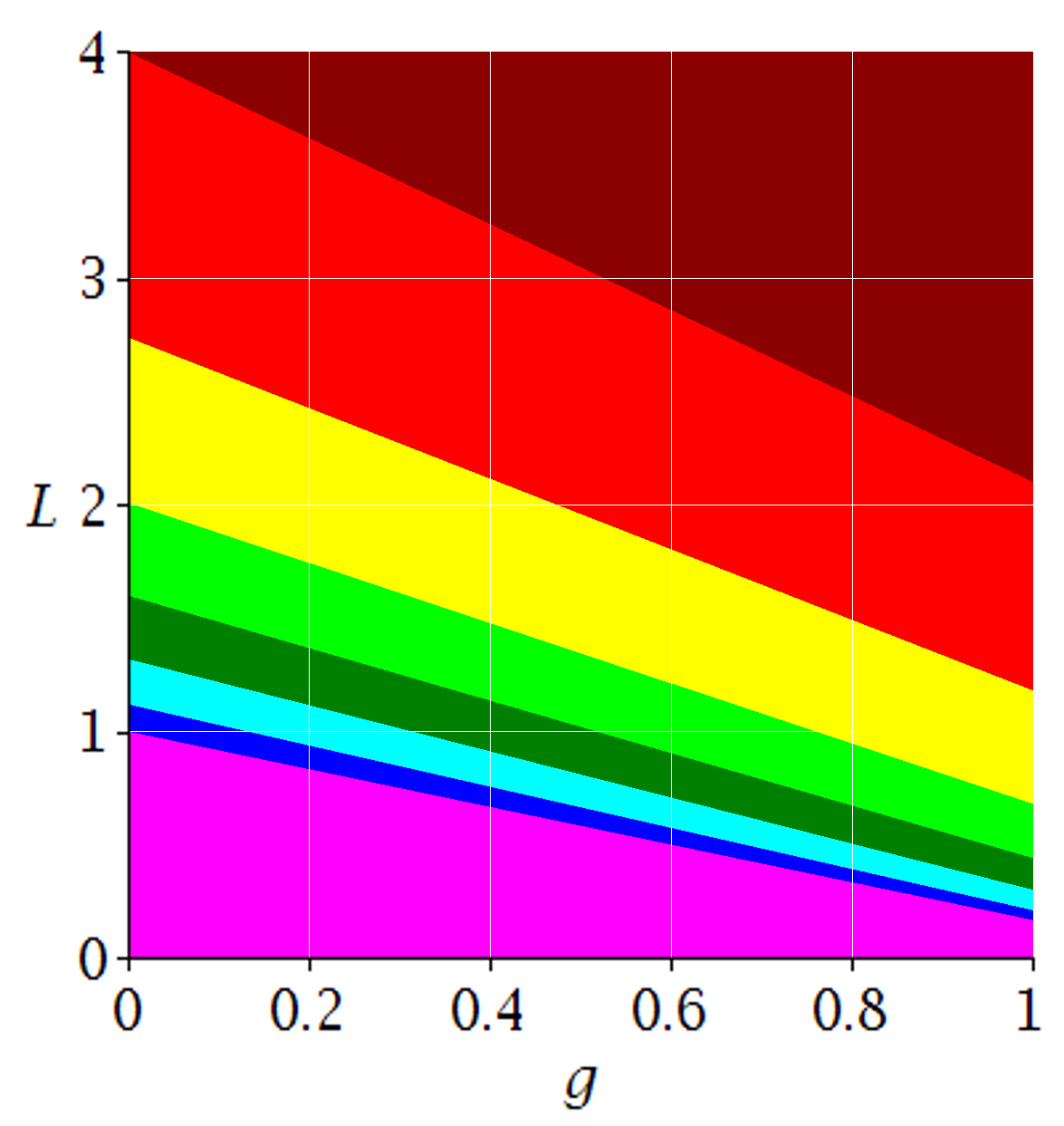}
\caption{The regions of Theorem \ref{thm-starterKeplerS} are represented in different colors. The uppermost region (in dark red) is actually unbounded}
\end{figure}

\begin{remark}
The exact solution $\widetilde{S}$ of the cubic equation $(1-g)\widetilde{S}+g\frac{\widetilde{S}^3}6 =L$ is given by
\[
\widetilde{S}(g,L) = \sqrt[3]{\frac{3L}{g}+\sqrt{\frac{9L^2}{g^2}+\frac{8(1-g)^3}{g^3}}} +\sqrt[3]{\frac{3L}{g}-\sqrt{\frac{9L^2}{g^2}+\frac{8(1-g)^3}{g^3}}},
\]
which can also be written as
\[
\widetilde{S}(g,L) = A-\frac{2(1-g)}{gA},
\]
where $A=\sqrt[3]{\frac{3L}{g}+\sqrt{\frac{9L^2}{g^2}+\frac{8(1-g)^3}{g^3}}} $.
\end{remark}

For any bounded region of $(0,1) \times [0,\infty)$ that excludes a neighborhood of $g=1, L=0$, we give in Section \ref{sec-const} an explicit construction of another piecewise-defined starter using only constants. This is important for building look-up table solutions of the hyperbolic Kepler's equation. Finally, we show that such a construction is impossible if a neighborhood of the corner is included.

\begin{thm}\label{thm-const-cubre}
For any $0<\varepsilon<\frac14$ and $L_{\max}>\varepsilon'= \frac{\sqrt{3}\alpha_0\varepsilon^\frac32}{(1-\varepsilon)^\frac12}$, there is a piecewise constant function $\widetilde{S}$ defined in $R=(0,1)\times [0,L_{\max}] \setminus (1-\varepsilon,1) \times [0,\varepsilon']$ that is an approximate zero of $f_{g,L}$.
\end{thm}

\begin{thm}\label{thm-const2}
Let $U \subseteq (0,1) \times [0,\infty)$ be an open set such that $\bar{U} \supseteq \{ 1 \} \times [0,\varepsilon]$ for some $\varepsilon>0$. For any constant starter $S_0\geq0$, there exists a point $(g_0,L_0) \in U$ such that $S_0$ is not an approximate zero of $f_{g_0,L_0}$.
\end{thm}

\section{Starters for the hyperbolic Kepler's equation}\label{sec-starter}

The aim of this section is to prove Theorem \ref{thm-starterKeplerS}. To do that, we first prove some necessary technical results. Then we show in Theorem \ref{thm-cubica-ng} that the solution of the cubic equation  $(1-g)\widetilde{S}+g\frac{\widetilde{S}^3}6 =L$ is an approximate zero when $0\leq L\leq 1-\frac56 g$. Finally, we study in Theorem \ref{thm-lineal} the family of starters $\widetilde{S}=L+ag$, and in Corollary \ref{cor-stripes} those with $a \in \{0.91, 1.02, 1.16, 1.33, 1.56, 1.90, 2.30 \}$.

\begin{lem}\label{lemma_asinh}
For all $k\geq1$, we have
\[
\frac{d^k}{dx^k}(\asinh (x))=(1+x^2)^{\frac12-k} P_k(x),
\]
where $P_k (x)$ are polynomials of degree $k-1$. The monomials of $P_k(x)$ have
all even or odd degree:
\[
  P_k(x)=a^{(k)}_{k-1}x^{k-1}+a^{(k)}_{k-3}x^{k-3}+a^{(k)}_{k-5}x^{k-5}+\cdots
\]
The leading coefficient is $(-1)^{k-1} (k-1)!$, the coefficients alternate signs, and the sum of
the absolute value of all the coefficients is $||P_k (x) ||_1 =1\cdot 3 \cdot 5\cdots (2k-3)=(2k-3)!!$
The independent term satisfies
\[
|P_k(0)|=
\begin{cases}
 ((k-2)!!)^2 & \text{ if $k$ is odd,}\\
 0           & \text{ if $k$ is even.}\\
\end{cases}
\]
\end{lem}
\begin{proof}
We proceed by induction. Note that $\asinh'(x)=(1+x^2)^{-\frac12}$, so the Lemma holds for $k=1$ by setting
$P_1(x)=1$. Assume now that the Lemma is true for some $k\geq 1$. Differentiating the $k$-th derivative, we get
\[
  \frac{d^{k+1}}{dx^{k+1}}\asinh(x)=(1+x^2)^{\frac12-(k+1)}\left[ (1-2k)xP_k(x)+(1+x^2)P'_k(x)\right],
\]
so we define $P_{k+1}(x)=(1-2k)xP_k(x)+(1+x^2)P'_k(x)$. This shows that $P_{k+1}(x)$ is a polynomial of degree
at most $k$. Moreover, if $P_k(x)$ has only terms of odd (or even) degree, then $P_{k+1}(x)$ has only terms of
even (or odd) degree, respectively. A more careful study of the leading coefficient of $P_{k+1}(x)$ from the
recurrence gives
\[
  a_{k}^{(k+1)}=(1-2k)a_{k-1}^{(k)}+(k-1)a_{k-1}^{(k)}=(-k)a_{k-1}^{(k)}=(-k)(-1)^{k-1}(k-1)!=(-1)^kk!,
\]
showing that the degree of $P_{k+1}(x)$ is $k$ and that the leading coefficient of $P_{k+1}(x)$ is
$(-1)^kk!$, as we needed for the inductive step. The other coefficients of $P_{k+1}(x)$ can be also found
by using the recurrence:
\[
  a_{k-2r}^{(k+1)}=(-k-2r)a_{k-2r-1}^{(k)}+(k-2r+1)a_{k-2r+1}^{(k)}
\]
for $r=1,\ldots,\left[\frac{k}2\right]$. Since $-k-2r<0$, $k-2r+1>0$, and the signs of $a_{k-2r-1}^{(k)}$ and
$a_{k-2r+1}^{(k)}$ are different, we have ${\rm sgn}(a_{k-2r}^{(k+1)})={\rm sgn}(a_{k-2r+1}^{(k)})$, which
alternate by the inductive hypothesis. Finally, for a polynomial with coefficients of alternating signs and
all even (or odd) exponents, we have $||P_{k+1}(x)||_1=|P_{k+1}(i)|=|(1-2k)iP_k(i)|=(2k-1)||P_k(x)||_1=(2k-1)!!$

According to~\cite{Jeffrey}, the power series expansion of $\asinh(x)$ at $x=0$ is
\[
  \asinh(x)=\sum_{n\geq0}\frac{(-1)^n(2n-1)!!}{(2n+1)(2n)!!}x^{2n+1} =\sum_{k\geq0} \frac{P_k(0)}{k!} x^k,
\]
so $P_{2n}(0)=0$, and $P_{2n+1}(0)=\frac{(-1)^n(2n-1)!!(2n+1)!}{(2n+1)(2n)!!}=(-1)^n(2n-1)!!^2$.
\end{proof}

\begin{lem}\label{lemma_sup}
Given a real number $x\neq0$ and an integer $n\geq2$, then
\[
\sup_{k\geq n} \left|\frac{x}{k} \right|^{\frac{1}{k-1}} =\max \left\{ 1, \left|\frac{x}{n} \right|^{\frac{1}{n-1}}  \right\}.
\]
\end{lem}
\begin{proof}
Assume first that $|x|\leq n$. Then $\max \{ 1,\left|\frac{x}{n} \right|^{\frac{1}{n-1}} \}=1$.
Moreover,
\[
\sup_{k\geq n} \left|\frac{x}{k} \right|^{\frac{1}{k-1}} \leq \sup_{k\geq n} \left|\frac{x}{n} \right|^{\frac{1}{k-1}} \leq 1
\]
and
\[
\sup_{k\geq n} \left|\frac{x}{k} \right|^{\frac{1}{k-1}} \geq \lim_{k \to \infty} \left|\frac{x}{k} \right|^{\frac{1}{k-1}}=
e^{\Lim{k \to \infty}\frac{\log|x|-\log|k|}{k-1} }=e^0=1.
\]
Therefore,
\[
\sup_{k\geq n} \left|\frac{x}{k} \right|^{\frac{1}{k-1}}=1=\max \left\{ 1, \left|\frac{x}{n} \right|^{\frac{1}{n-1}}  \right\}.
\]

Now consider the case $|x| >n$. Then $\max \left\{ 1, \left|\frac{x}{n} \right|^{\frac{1}{n-1}}  \right\}=\left|\frac{x}{n} \right|^{\frac{1}{n-1}}$. Moreover,
\[
\sup_{k\geq n} \left|\frac{x}{k} \right|^{\frac{1}{k-1}} \geq \left|\frac{x}{n} \right|^{\frac{1}{n-1}}
\]
since the supremum is greater than or equal to the first term, and
\[
\sup_{k\geq n} \left|\frac{x}{k} \right|^{\frac{1}{k-1}} \leq \sup_{k\geq n} \left|\frac{x}{n} \right|^{\frac{1}{k-1}}
\leq \left|\frac{x}{n} \right|^{\frac{1}{n-1}}
\]
since the sequence $\left|\frac{x}{n} \right|^{\frac{1}{n-1}} $ is decreasing.
Therefore,
\[
\sup_{k\geq n} \left|\frac{x}{k} \right|^{\frac{1}{k-1}} =
\left|\frac{x}{n} \right|^{\frac{1}{n-1}}= \max \left\{ 1, \left|\frac{x}{n} \right|^{\frac{1}{n-1}}  \right\},
\]
which concludes the proof.
\end{proof}

The previous two lemmas give us the following upper bound for $\gamma(f_{g,L},\widetilde{S})$.
\begin{lem}\label{lemma_gamma}
For all $n\geq2$, and $(g,L) \in (0,1) \times [0,\infty)$,
\[
\gamma(f_{g,L},\widetilde{S}) \leq
\frac{1}{1+\widetilde{S}^2} \max \left\{a_2, \ldots, a_{n-1},
2\max\{1,|\widetilde{S}|\} \max\left\{ 1,\left( \frac{g}{n(\sqrt{1+\widetilde{S}^2}-g)}\right)^{\frac{1}{n-1}} \right\} \right\},
\]
where
\[
a_k=\left( \frac{g|P_k(\widetilde{S})|}{k! (\sqrt{1+\widetilde{S}^2}-g)} \right)^{\frac{1}{k-1}}.
\]
\end{lem}
\begin{proof}
By Lemma \ref{lemma_asinh},
\begin{align*}
\gamma(f_{g,L},\widetilde{S}) &=\sup_{k \geq 2} \left| \frac{f_{g,L}^{(k)}(\widetilde{S})}{k!f_{g,L}'(\widetilde{S})}\right|^{\frac{1}{k-1}}
=\sup_{k \geq 2} \left| \frac{g P_k (\widetilde{S}) (1+\widetilde{S}^2)^{\frac12-k}}{k! \left(1-\frac{g}{\sqrt{1+\widetilde{S}^2}}\right)}\right|^{\frac{1}{k-1}}\\
&=\sup_{k \geq 2} \left| \frac{g P_k (\widetilde{S}) (1+\widetilde{S}^2)^{1-k}}{k! \left( \sqrt{1+\widetilde{S}^2}-g\right) } \right|^{\frac{1}{k-1}}
=\frac{1}{1+\widetilde{S}^2}\sup_{k \geq 2} \left| \frac{g P_k (\widetilde{S})}{k! \left( \sqrt{1+\widetilde{S}^2}-g\right) } \right|^{\frac{1}{k-1}}\\
&=\frac{1}{1+\widetilde{S}^2} \max \left\{a_2,\ldots,a_{n-1}, \sup_{k \geq n} a_k \right\}.
\end{align*}
Again by Lemma \ref{lemma_asinh},
\begin{align*}
a_k& \leq \left| \frac{g \max\{1,|\widetilde{S}| \}^{k-1} \cdot (2k-3)!!}{k! \left( \sqrt{1+\widetilde{S}^2}-g\right) } \right|^{\frac{1}{k-1}}
\leq \max\{1,|\widetilde{S}| \} \left| \frac{g  \cdot (2k-2)!!}{k! \left( \sqrt{1+\widetilde{S}^2}-g\right) } \right|^{\frac{1}{k-1}}\\
&= \max\{1,|\widetilde{S}| \}  \left| \frac{g \cdot 2^{k-1}}{k\left(\sqrt{1+\widetilde{S}^2} -g\right)}   \right|^{\frac{1}{k-1}}
=2\max\{1,|\widetilde{S}| \}   \left| \frac{g }{k\left(\sqrt{1+\widetilde{S}^2} -g\right)}  \right|^{\frac{1}{k-1}}.
\end{align*}
By Lemma \ref{lemma_sup},
\begin{align*}
\sup_{k\geq n} a_k
&\leq 2\max\{1,|\widetilde{S}| \} \sup_{k\geq n} \left| \frac{g }{k\left(\sqrt{1+\widetilde{S}^2} -g\right)}  \right|^{\frac{1}{k-1}}\\
&=2\max\{1,|\widetilde{S}| \} \max \left\{ 1, \left( \frac{g }{n \left(\sqrt{1+\widetilde{S}^2} -g\right)}  \right)^{\frac{1}{n-1}} \right\}.
\end{align*}
\end{proof}


\begin{thm}\label{thm-cubica-ng}
The exact solution $\widetilde{S}$ of the cubic equation $(1-g)\widetilde{S}+g\frac{\widetilde{S}^3}6 =L$ is an approximate zero of $f_{g,L}$ in the region $R_1$, where
\[
R_{1} =\left\{ 0 < g < 1, \; 0 \leq L \leq 1-\frac{5}{6}g  \right\}.
\]
\end{thm}
\begin{proof}
For any point $(g,L) \in R_1$, the solution $\widetilde{S}(g,L)$ of the cubic equation satisfies $\widetilde{S} \in [0,1]$. Therefore,
\begin{align*}
| f_{g,L}(\widetilde{S})|  &= \left| \widetilde{S} -g \asinh(\widetilde{S}) -L \right|
=\left| \left((1-g) \widetilde{S}+\frac{g}{6} \widetilde{S}^3 -L \right) -  g \left( \asinh(\widetilde{S}) -\widetilde{S}+\frac16 \widetilde{S}^3\right) \right|\\
&=g \left(  \asinh(\widetilde{S}) -\widetilde{S}+\frac16 \widetilde{S}^3 \right) \leq
\asinh(\widetilde{S}) -\widetilde{S}+\frac16 \widetilde{S}^3.
\end{align*}

Using the previous inequality in the definition of $\beta(f_{g,L},\widetilde{S})$ and the fact that $g<1$, we obtain
\[
\beta(f_{g,L},\widetilde{S})= \left| \frac{f(\widetilde{S})}{f'(\widetilde{S})}   \right| \leq
\frac{\asinh(\widetilde{S}) -\widetilde{S}+\frac16 \widetilde{S}^3}{1-\frac{1}{\sqrt{1+\widetilde{S}^2}}}=
\frac{\left(\asinh(\widetilde{S}) -\widetilde{S}+\frac16 \widetilde{S}^3 \right) \sqrt{1+\widetilde{S}^2}}{ \sqrt{1+\widetilde{S}^2} -1}.
\]

On the other hand, it follows from Lemma \ref{lemma_gamma} for $n=3$, $g<1$ and $|\widetilde{S}| \leq1$ that
\begin{align*}
\gamma(f_{g,L},\widetilde{S}) &\leq \frac{1}{1+\widetilde{S}^2} \max \left\{ \frac{\widetilde{S}}{2 (\sqrt{1+\widetilde{S}^2}-1 )},
2 \max \left\{ 1, \left( \frac{1}{3(\sqrt{1+\widetilde{S}^2}-1)} \right)^{\frac12} \right\} \right\}\\
&=\frac{1}{1+\widetilde{S}^2}  \max \left\{ 2, \frac{2}{\sqrt{3} (\sqrt{1+\widetilde{S}^2}-1)^{\frac12} } \right\},
\end{align*}
since $\frac{\widetilde{S}}{2 \left(\sqrt{1+\widetilde{S}^2}-1 \right)}  \leq \frac{2}{\sqrt{3} \left(\sqrt{1+\widetilde{S}^2}-1\right)^{\frac12} }$.

Multiplying the inequalities for $\beta(f_{g,L},\widetilde{S})$ and $\gamma(f_{g,L},\widetilde{S})$ together, we have that the $\alpha$-test is satisfied if
\[
\frac{\left(\asinh(\widetilde{S}) -\widetilde{S}+\frac16 \widetilde{S}^3 \right) }{\sqrt{1+\widetilde{S}^2} \left( \sqrt{1+\widetilde{S}^2} -1 \right)}
 \max \left\{ 2, \frac{2}{\sqrt{3} (\sqrt{1+\widetilde{S}^2}-1)^{\frac12} } \right\}
< \alpha_0.
\]
A standard analytic study of the previous one-variable function shows that it is bounded above by $\alpha_0$ in the interval $\widetilde{S} \in [0,1]$.
\end{proof}

\begin{thm}\label{thm-lineal}
The starter $\widetilde{S}(g,L)=L+a \,g$ is an approximate zero of $f_{g,L}$ in the stripe
\[
R_2 (a) =\{0< g < 1, \widetilde{S}_{\min}(a) \leq L +a g \leq \widetilde{S}_{\max} (a) \}
\]
if $\widetilde{S}_{\min}(a),\widetilde{S}_{\max}(a)$ are chosen such that
\[
\frac{ |a-\asinh(x)|}{\sqrt{1+x^2}(\sqrt{1+x^2}- 1)^2}
\max\left\{ 1, x , 2x(\sqrt{1+x^2}-1) \right\}<\alpha_0
\]
for all $x\in [\widetilde{S}_{\min}(a),\widetilde{S}_{\max}(a)]$.
\end{thm}
\begin{proof}
By definition, we have
\[
\beta(f_{g,L},\widetilde{S})=\left|\frac{f_{g,L}(\widetilde{S})}{f_{g,L}'(\widetilde{S})} \right|
=\frac{g |a-\asinh(\widetilde{S})|}{1-\frac{g}{\sqrt{1+\widetilde{S}^2}}}
\leq \frac{ |a-\asinh(\widetilde{S})|}{1-\frac{1}{\sqrt{1+\widetilde{S}^2}}}.
\]

By Lemma~\ref{lemma_gamma} with $n=2$,
\[
\gamma(f_{g,L},\widetilde{S}) \leq \frac{2}{1+\widetilde{S}^2} \max\{1,\widetilde{S} \} \max\left\{1, \frac{1}{2(\sqrt{1+\widetilde{S}^2}-1)} \right\}.
\]

Therefore, we obtain that
\begin{align*}
\alpha(f_{g,L},\widetilde{S}) &=\beta(f_{g,L},\widetilde{S}) \cdot \gamma(f_{g,L},\widetilde{S}) \\
&\leq \frac{2 |a-\asinh(\widetilde{S})|}{\sqrt{1+\widetilde{S}^2}(\sqrt{1+\widetilde{S}^2}- 1)} \max\{1,\widetilde{S} \} \max\left\{1, \frac{1}{2(\sqrt{1+\widetilde{S}^2}-1)} \right\}\\
&= \frac{ |a-\asinh(\widetilde{S})|}{\sqrt{1+\widetilde{S}^2}(\sqrt{1+\widetilde{S}^2}- 1)^2}
\max\left\{ 1, \widetilde{S} , 2\widetilde{S}(\sqrt{1+\widetilde{S}^2}-1) \right\},
\end{align*}
which is less than $\alpha_0$ in $R_2(a)$ when $\widetilde{S} \in [\widetilde{S}_{\min}(a),\widetilde{S}_{\max}(a)]$.
\end{proof}

\begin{cor}\label{cor-stripes}
The starter $\widetilde{S}(g,L)=L+a \,g$ is an approximate zero of $f_{g,L}$ in the stripe
\[
R_2 (a) =\{0< g < 1, \widetilde{S}_{\min}(a) \leq L +a g \leq \widetilde{S}_{\max} (a) \}
\]
for the values of $a$, $\widetilde{S}_{\min}(a)$, $\widetilde{S}_{\max}(a)$ given in the following table:
\[
\begin{tabular}{c|c|c}
   a   & $\widetilde{S}_{\min}(a)$ & $\widetilde{S}_{\max}(a)$ \\
\hline
$0.91$ &   $0.99$              & $1.12$\\
$1.02$ &   $1.12$              & $1.32$\\
$1.16$ &   $1.32$              & $1.60$\\
$1.33$ &   $1.59$              & $2.01$\\
$1.56$ &   $2.00$              & $2.74$\\
$1.90$ &   $2.73$              & $4.00$\\
$2.30$ &   $4.00$              & $\infty$
\end{tabular}
\]
\end{cor}

\begin{proof}[Proof of Theorem~\ref{thm-starterKeplerS}]
By Theorem \ref{thm-cubica-ng}, $\widetilde{S}(g,L)$ is an approximate zero of $f_{g,L}$ when $0\leq L \leq 1-\frac56 g$. For the remaining cases, we use Corollary \ref{cor-stripes}. It is easy to verify that for each $a \in \{ 0.91, 1.02, 1.16, 1.33, 1.56, 1.90, 2.30\}$, the region where $\widetilde{S}$ is defined as $L+a \, g$ is contained in $R_2(a)$.
\end{proof}

\section{An analytical study of piecewise constant starters}\label{sec-const}

\begin{thm}\label{thm-const}
The constant starter $\widetilde{S}(g,L)=S_0 >0$ is an approximate zero of $f_{g,L}$ in the stripe
\[
R_3(S_0)=\left\{ 0<g<1, S_0-g\asinh(S_0) -\Delta (S_0) <L <
S_0-g\asinh(S_0) +\Delta(S_0) \right\},
\]
where
\[
\Delta (S_0)=
\begin{cases}
\frac{\sqrt{3}\alpha_0 S_0^3 \sqrt{1+S_0^2}}{2(\sqrt{1+S_0^2}+1)^{\nicefrac32}},    & \text{ if } 0 < S_0 \leq \frac{\sqrt7}3,\\[5mm]
\frac{\alpha_0 S_0 \min\{1,S_0\} \sqrt{1+S_0^2}}{2(\sqrt{1+S_0^2}+1)},              &  \text{ if } S_0 \geq \frac{\sqrt7}3.
\end{cases}
\]
The constant starter $\widetilde{S}(g,L)=0$ is an approximate zero of $f_{g,L}$ in the region
\[
R_3 (0)=\left\{ 0<g<1, 0<L < \min \left\{ \alpha_0 (1-g), \frac{\sqrt3 \alpha_0 (1-g)^\frac32 }{g^\frac12} \right\}  \right\}.
\]
\end{thm}
\begin{proof}
Consider first the case $S_0>0$. By Definition \ref{def} and the fact that $g<1$,
\begin{align*}
\beta(f_{g,L},\widetilde{S}) &=\left|\frac{f_{g,L}(\widetilde{S})}{f_{g,L}'(\widetilde{S})} \right|
=\left|\frac{S_0-g\asinh(S_0)-L}{1-\frac{g}{\sqrt{1+S_0^2}}} \right| \\
&\leq \frac{\left|S_0-g\asinh(S_0)-L\right|}{1-\frac{1}{\sqrt{1+S_0^2}}} =\frac{1}{S_0^2}\left|S_0-g\asinh(S_0)-L\right| \sqrt{1+S_0^2} \left(\sqrt{1+S_0^2}+1 \right).
\end{align*}
On the other hand, it follows from Lemma~\ref{lemma_gamma} with $n=3$ that
\[
\begin{aligned}
\gamma(f_{g,L},\widetilde{S})
&\leq
\frac{1}{1+S_0^2} \max \left\{
\frac{ S_0}{2 (\sqrt{1+S_0^2}-1)},
2\max\{ 1,S_0 \} \max\left\{ 1,\left( \frac{1}{3(\sqrt{1+S_0^2}-1)}\right)^{\frac{1}{2}} \right\} \right\}\\
&=\frac{2}{1+S_0^2} \max\{ 1,S_0 \} \max\left\{ 1,\left( \frac{1}{3(\sqrt{1+S_0^2}-1)}\right)^{\frac{1}{2}} \right\}.
\end{aligned}
\]

We will now distinguish between $S_0\geq1$ and $0 < S_0 < 1$. In the first case, $\gamma(f_{g,L},\widetilde{S}) \leq \frac{2S_0}{1+S_0^2}$, so the $\alpha$-test follows from
\[
\frac{2\left(\sqrt{1+S_0^2}+1 \right) \left|S_0-g\asinh(S_0)-L\right|}{S_0\sqrt{1+S_0^2}}
 <\alpha_0,
\]
which is equivalent to
\begin{align*}
\left|S_0-g\asinh(S_0)-L \right|
&<
\frac{\alpha_0 S_0 \sqrt{1+S_0^2}}{2\left(\sqrt{1+S_0^2}+1 \right)}=\Delta(S_0).
\end{align*}

When $0 < S_0 < 1$, $\gamma(f_{g,L},\widetilde{S}) \leq  \frac{2}{1+S_0^2} \max\left\{ 1,\left( \frac{1}{3(\sqrt{1+S_0^2}-1)}\right)^{\frac{1}{2}} \right\}$ so
the $\alpha$-test is satisfied if
\[
\left|S_0-g\asinh(S_0)-L\right|
\frac{2\left(\sqrt{1+S_0^2}+1 \right)}{S_0^2\sqrt{1+S_0^2}} \max\left\{ 1,\left( \frac{1}{3(\sqrt{1+S_0^2}-1)}\right)^{\frac{1}{2}} \right\}
<\alpha_0,
\]
which is equivalent to
\begin{align*}
\left|S_0-g\asinh(S_0)-L \right|
&<
\min \left\{
\frac{\alpha_0 S_0^2 \sqrt{1+S_0^2}}{2 \left(\sqrt{1+S_0^2}+1 \right)},
\frac{\sqrt{3} \alpha_0 S_0^3 \sqrt{1+S_0^2}}{2\left(\sqrt{1+S_0^2}+1 \right)^\frac32}
\right\}\\
&=
\left\{
\begin{array}{cl}
\frac{\sqrt{3}\alpha_0 S_0^3 \sqrt{1+S_0^2}}{2(\sqrt{1+S_0^2}+1)^{\nicefrac32}} & \text{ if } 0 < S_0 \leq \frac{\sqrt7}3
\\[5mm]
\frac{\alpha_0 S_0^2 \sqrt{1+S_0^2}}{2(\sqrt{1+S_0^2}+1)} & \text{ if } \frac{\sqrt7}3 \leq S_0 <1
\end{array}
\right\}
=\Delta(S_0).
\end{align*}
Therefore, the $\alpha$-test is satisfied in the whole region $R_3 (S_0)$.

When $\widetilde{S}(g,L)=S_0=0$, we have
\[
\beta(f_{g,L},\widetilde{S})=\left|\frac{f_{g,L}(\widetilde{S})}{f_{g,L}'(\widetilde{S})} \right|=\frac{L}{1-g}
\]
and by Lemmas \ref{lemma_asinh} and \ref{lemma_sup}
\[
\begin{aligned}
\gamma(f_{g,L},\widetilde{S}) &= \sup_{k\geq2} \left| \frac{f^{(k)}(0)}{k! (1-g)}\right|^{\frac1{k-1}}
=\sup_{k\geq2} \left| \frac{g P_k(0)}{k! (1-g)}\right|^{\frac1{k-1}}
=\sup_{\doslineas{k\geq3}{\text{odd}}} \left| \frac{g ((k-2)!!)^2}{k! (1-g)}\right|^{\frac1{k-1}} \\
&\leq \sup_{k\geq3} \left| \frac{g}{k (1-g)}\right|^{\frac1{k-1}} = \max \left\{ 1, \sqrt{\frac{g}{3(1-g)}} \right\}.
\end{aligned}
\]
Therefore, the $\alpha$-test is satisfied if
\[
\frac{L}{1-g} \max \left\{ 1, \sqrt{\frac{g}{3(1-g)}} \right\} < \alpha_0,
\]
which is equivalent to $(g,L) \in R_3(0)$.
\end{proof}

\begin{remark}\label{remark-const}
The region $R_3(0)$ of Theorem \ref{thm-const} can be written as
\[
R_3(0)= \left\{ 0 < g < \frac34, 0\leq L < \alpha_0 (1-g) \right\} \cup  \left\{ \frac34 \leq g < 1, 0\leq L <  \frac{\sqrt3\alpha_0 (1-g)^\frac32}{g^\frac12} \right\}.
\]
\end{remark}

\begin{figure}[!ht]
\centering
\includegraphics[height=5cm]{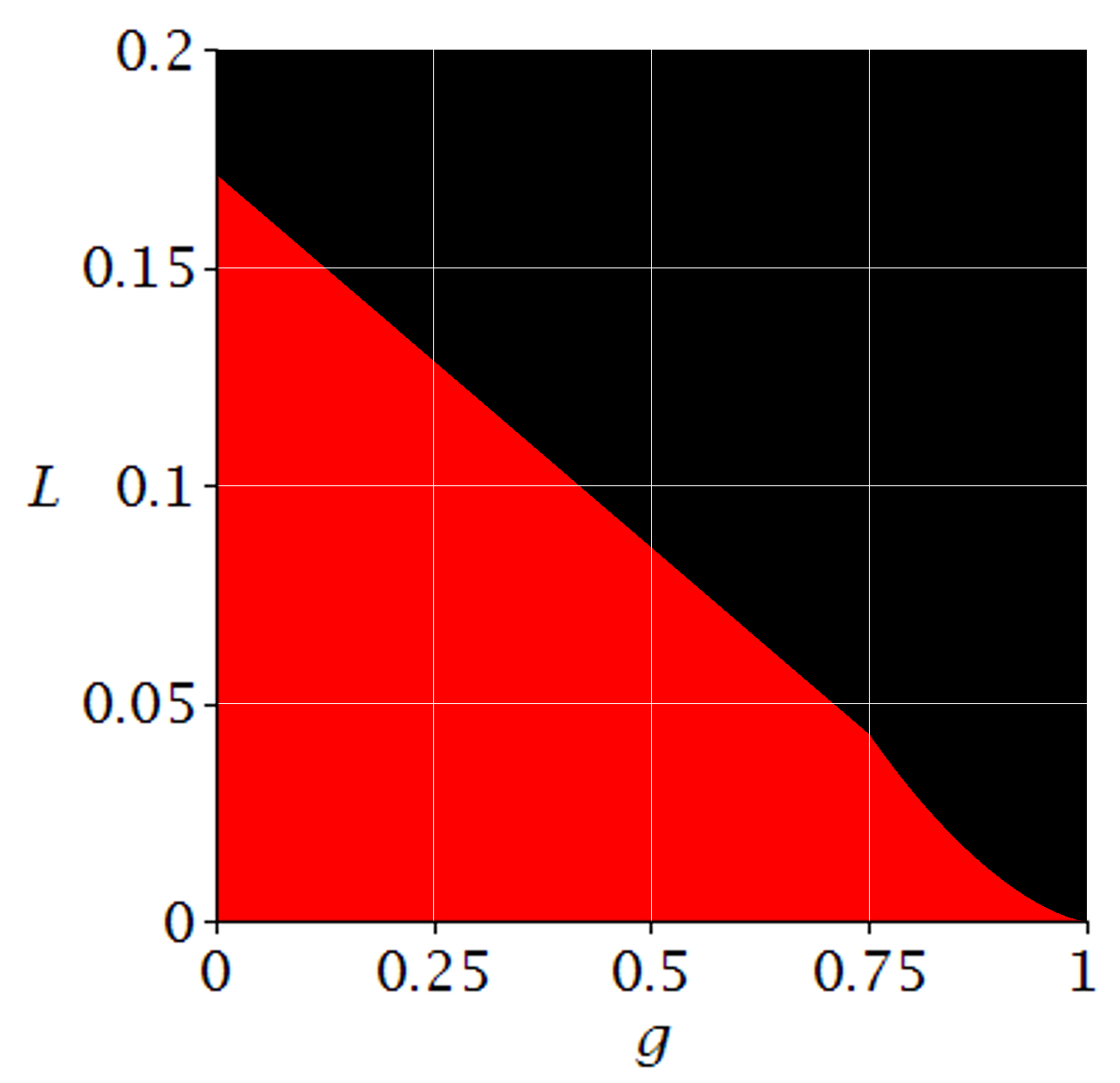}
\caption{The region $R_3(0)$ of Theorem \ref{thm-const} and Remark \ref{remark-const} is represented in red}
\end{figure}

\begin{lem}\label{lem-delta}
The function $\Delta(S_0)$ of Theorem \ref{thm-const} is strictly increasing for $S_0>0$.
\end{lem}
\begin{proof}
When $0<S_0 \leq \frac{\sqrt7}3$, then
\[
\Delta (S_0)= \frac{\sqrt{3}\alpha_0 S_0^3 \sqrt{1+S_0^2}}{2(\sqrt{1+S_0^2}+1)^{\nicefrac32}}
= \frac{\sqrt{3}\alpha_0}2  S_0^\frac52 \left(1-\frac{1}{1+S_0^2} \right)^\frac14 \left(  1-\frac{1}{\sqrt{1+S_0^2}+1} \right)^\frac32.
\]
When $S_0 > \frac{\sqrt7}3$, then
\[
\Delta(S_0)=\frac{\alpha_0 S_0 \min\{1,S_0\} \sqrt{1+S_0^2}}{2(\sqrt{1+S_0^2}+1)}
=\frac{\alpha_0}2   S_0 \min\{1,S_0\} \left(  1-\frac{1}{\sqrt{1+S_0^2}+1} \right).
\]
Since both expressions are the product of constant and strictly increasing functions and $\Delta(S_0)$ is continuous at $S_0=\frac{\sqrt7}3$, then $\Delta(S_0)$ is strictly increasing.
\end{proof}

\begin{proof}[Proof of Theorem~\ref{thm-const-cubre}]
We decompose the region of the theorem as
\[
\underbrace{(0,1-\varepsilon) \times [0,\varepsilon']}_{A}
\cup \underbrace{(0,1) \times [\varepsilon',L_{\max}]}_{B}
\]
From Remark \ref{remark-const}, we have that $A \subseteq R_3 (0)$.

Define the sequence $S_0=\varepsilon'$ and $S_{i+1}=S_i +2\Delta(S_i)$, for $i\geq0$. Since $\Delta$ is a positive and strictly increasing function by Lemma \ref{lem-delta}, the sequence $S_i$ is strictly increasing and satisfies
$S_i \geq S_0+2 i \Delta(S_0) \longrightarrow\infty$ as $i \to \infty$.

Since $\displaystyle{\lim_{x \to \infty}} (x-\asinh(x))=\infty$, there exists a large enough $N$ for which $S_N-\asinh(S_N) >L_{\max}$. Once we prove that $\displaystyle{B \subseteq C=\bigcup_{i=0}^N R_3 (S_i)}$, we will have that $A \cup B$ is covered by the regions where the starters $0,S_0,\ldots,S_N$ are approximate zeros of $f_{g,L}$. In particular, this will provide the piecewise constant starter $\widetilde{S}$ that we looked for. 

It only remains to prove that $B \subseteq C$. Note first that $R_3(S_i)$ are stripes whose union is
\[
  C = \left\{0<g<1, S_0-\Delta(S_0)-g\asinh(S_0)<L<S_N+\Delta(S_N)-g\asinh(S_N) \right\}
\]
since the lower boundary of $R_3(S_{i+1})$ is a segment that is always below the upper boundary of $R_3(S_i)$. Indeed, it is enough to see that the endpoints of the first segment, $(0,S_{i+1}-\Delta(S_{i+1}))$ and $(1,S_{i+1}-\Delta(S_{i+1})-\asinh(S_{i+1}))$, are below the endpoints of the second segment, $(0,S_i+\Delta(S_i))$ and $(1,S_i+\Delta(S_i)-\asinh(S_i))$, respectively:
\[
\begin{aligned}
  &S_{i+1}-\Delta(S_{i+1})<S_i+\Delta(S_i) \iff \Delta(S_i)<\Delta(S_{i+1}),\\
  &S_{i+1}-\Delta(S_{i+1})-\asinh(S_{i+1})<S_i+\Delta(S_i)-\asinh(S_i) \iff \\ &\qquad\iff\Delta(S_i)+\asinh(S_i)<\Delta(S_{i+1})+\asinh(S_{i+1}),
\end{aligned}
\]
which is true because $\Delta$ and $\asinh$ are both strictly increasing.

On the other hand, the vertices of the rectangle $B$ are in $\bar{C}$, which means that $B\subseteq C$. Indeed, it follows from $\asinh(\varepsilon'),\asinh(S_N),\Delta(\varepsilon'),\Delta(S_N)>0$; $\varepsilon'<L_{\max}$ by hypothesis and $S_N-\asinh(S_N)>L_{\max}$ by construction, that
\begin{align*}
  \varepsilon'-\Delta(\varepsilon')-&\asinh(\varepsilon')<  \varepsilon'-\Delta(\varepsilon')<   \varepsilon'<L_{\max}<\\
  &<S_N-\asinh(S_N)<S_N+\Delta(S_N)-\asinh(S_N)<S_N+\Delta(S_N).
\end{align*}
This implies
\begin{align*}
  &\varepsilon'-\Delta(\varepsilon')<\varepsilon'<S_N+\Delta(S_N)                                   \Rightarrow  (0,\varepsilon')\in\bar{C}, \\
  &\varepsilon'-\Delta(\varepsilon')-\asinh(\varepsilon')<\varepsilon'<S_N+\Delta(S_N)-\asinh(S_N)   \Rightarrow (1,\varepsilon')\in\bar{C},\\
  & \varepsilon'-\Delta(\varepsilon')<L_{\max}<S_N+\Delta(S_N)                                  \Rightarrow    (0,L_{\max})\in\bar{C},\\      &\varepsilon'-\Delta(\varepsilon')-\asinh(\varepsilon')<L_{\max}<S_N+\Delta(S_N)-\asinh(S_N)  \Rightarrow    (1,L_{\max})\in\bar{C},
\end{align*}
concluding the proof.
\end{proof}

\begin{proof}[Proof of Theorem~\ref{thm-const2}]
We proceed by contradiction, i.e. we assume that $S_0$ is an approximate zero of $f_{g,L}$ for all $(g,L)\in U$.
By definition, this means that
\begin{equation}\label{eq1}
  \left|\frac{S_0-g\asinh(S_0)-L}{\sqrt{1+S_0^2} \left(\sqrt{1+S_0^2}-g\right)} \right|\,
  \sup_{k\geq2}\left|\frac{gP_k(S_0)}{k!\left(\sqrt{1+S_0^2}-g\right)} \right|^\frac1{k-1}<\alpha_0
\end{equation}
for all $(g,L)\in U$. We will derive a contradiction from \eqref{eq1} for all $S_0\geq 0$.

Case $S_0=0$. Inequality \eqref{eq1} is equivalent to
\[
  \frac{L}{1-g}\sup_{k\geq2}\left|\frac{gP_k(0)}{k!(1-g)} \right|^\frac1{k-1}<\alpha_0.
\]
Using Lemma~\ref{lemma_asinh}, we obtain
\[
\alpha_0>\frac{L}{1-g}\sup_{\doslineas{k\geq3}{\text{odd}}}\left|\frac{g ((k-2)!!)^2}{k!(1-g)}\right|^\frac1{k-1}
  \geq \frac{L}{1-g}\sup_{\doslineas{k\geq3}{\text{odd}}}\left|\frac{g}{k(k-1)(1-g)}\right|^\frac1{k-1}\geq
  \frac{Lg^\frac12}{\sqrt{6}(1-g)^\frac32}.
\]
This implies that $L\leq\frac{\sqrt{6}\alpha_0(1-g)^\frac32}{g^\frac12}$ in $\bar{U}$. Therefore $L\leq 0$ in
$\bar{U}\cap\{g=1\}=\{1\}\times[0,\varepsilon]$, which is impossible.

For the rest of the cases, we take limit as $g \to 1^-$ and then limit as $L \to 0^+$, obtaining
\begin{equation}\label{eq3}
\left|\frac{S_0-\asinh(S_0)}{\sqrt{1+S_0^2}\left(\sqrt{1+S_0^2}-1\right)} \right|\,
  \sup_{k\geq2}\left|\frac{P_k(S_0)}{k!\left(\sqrt{1+S_0^2}-1\right)} \right|^\frac1{k-1} \leq \alpha_0.
\end{equation}

Case $0<S_0<1.598$. Inequality \eqref{eq3} implies that
\[
\alpha_0 \geq \frac{(S_0-\asinh(S_0)) (\sqrt{1+S_0^2}+1)^2}{2S_0^3 \sqrt{1+S_0^2}} \geq
\frac{2(S_0-\asinh(S_0))}{S_0^3 }
\]
because $(t+1)^2/t \geq 4$ for all $t>0$.
Therefore, $h(S_0)=\asinh(S_0) -S_0+\frac{\alpha_0S_0^3}{2} \geq0$, which is false in the interval $(0,1.598)$ because $h$ has only one critical point there, which is a minimum, and $h(0),h(1.598)\leq0$.

Case $1.598\leq S_0<3.1$. Inequality \eqref{eq3} is equivalent to
\[
  \frac{(S_0-\asinh(S_0))(\sqrt{1+S_0^2}+1)^\frac32 \sqrt{2S_0^2-1}}{\sqrt6 S_0^3 \sqrt{1+S_0^2}} \leq\alpha_0.
\]
Since $\frac{(\sqrt{1+S_0^2}+1)^\frac32 \sqrt{2S_0^2-1}}{S_0 \sqrt{1+S_0^2}}$ is increasing in this interval, then
\[
 \alpha_0 
 \geq \frac{S_0-\asinh(S_0)}{\sqrt6 S_0^2} \cdot \frac{(\sqrt{1+1.598^2}+1)^\frac32 \sqrt{2 \cdot 1.598^2-1}}{1.598 \sqrt{1+1.598^2}}.
\]
Therefore,
\[
\frac{S_0-\asinh(S_0)}{S_0^2} \leq\frac{\sqrt{6}\alpha_0 1.598 \sqrt{1+1.598^2}}{(\sqrt{1+1.598^2}+1)^\frac32 \sqrt{2 \cdot 1.598^2-1}} <0.13,
\]
and thus $S_0-\asinh(S_0)-0.13S_0^2<0$. This is false since the function $h(S_0)=S_0-\asinh(S_0)-0.13S_0^2$ has an absolute minimum at $S_0=3.1$ and  $h(3.1)>0$.

Case $3.1\leq S_0<9.62$. Inequality \eqref{eq3} is equivalent to
\begin{equation}\label{eq2}
  \frac{(S_0-\asinh(S_0))(\sqrt{1+S_0^2}+1)^\frac43 (6S_0^2-9)^\frac13}{\sqrt[3]{24} S_0^\frac73  \sqrt{1+S_0^2}} \leq\alpha_0.
\end{equation}
On the one hand,
\[
\frac{(\sqrt{1+S_0^2}+1)^\frac43 (6S_0^2-9)^\frac13}{S_0  \sqrt{1+S_0^2}} \geq \frac{(\sqrt{1+9.62^2}+1)^\frac43 (6\cdot 9.62^2-9)^\frac13}{9.62 \sqrt{1+9.62^2}} > 2.064,
\]
since the expression on the left is decreasing. On the other hand, $\frac{S_0-\asinh(S_0)}{S_0^\frac43} > 0.24$ because $h(S_0)=S_0-\asinh(S_0)-0.24S_0^\frac43$ is increasing in the interval, so $h(S_0)\geq h(3.1)>0$.

Substituting these inequalities in Eq.~\eqref{eq2}, we obtain
\[
\alpha_0 \geq \frac{(S_0-\asinh(S_0))(\sqrt{1+S_0^2}+1)^\frac43 (6S_0^2-9)^\frac13}{\sqrt[3]{24} S_0^\frac73  \sqrt{1+S_0^2}} > \frac{2.064 \cdot 0.24}{\sqrt[3]{24}}>0.1717>\alpha_0,
\]
which is a contradiction.

Case $S_0\geq 9.62$. Using that $\frac{S_0-\asinh(S_0)}{\sqrt{1+S_0^2}}\geq\frac{1}{2}$  and
that $\sqrt{1+S_0^2}-1\leq S_0$ in the interval, we obtain from Eq.~\eqref{eq3} that
\[
\alpha_0
\geq \frac{1}{2\left(\sqrt{1+S_0^2}-1\right)}
  \sup_{k\geq2}\left|\frac{P_k(S_0)}{k!\left(\sqrt{1+S_0^2}-1\right)} \right|^\frac1{k-1}
\geq \frac{1}{2S_0} \sup_{k\geq2}\left|\frac{P_k(S_0)}{k!S_0} \right|^\frac1{k-1}.
\]
This implies that, for all $k\geq 2$,
\[
  |P_k(S_0)|\leq (2\alpha_0)^{k-1}k!S_0^k.
\]
By Lemma \ref{lemma_asinh}, the leading term of $P_k(S_0)$ is $\pm (k-1)!S_0^{k-1}$, its coefficients add up
to a maximum of $(2k-3)!!$ and it has no term of degree $k-2$, so
\[
  (2\alpha_0)^{k-1}k!S_0^k\geq |P_k(S_0)|\geq(k-1)!S_0^{k-1}-(2k-3)!!S_0^{k-3}, \text{ for all $k\geq 2$},
\]
which is equivalent to
\[
  (2\alpha_0)^{k-1}kS_0^3-S_0^2+\frac{(2k-3)!!}{(k-1)!}\geq 0, \text{ for all $k\geq 2$}.
\]
This implies that
\[
  h_k(S_0)=(2\alpha_0)^{k-1}kS_0^3-S_0^2+2^{k-1}\geq 0, \text{ for all $k\geq 2$}.
\]
The degree $3$ polynomial $h_k$ has two local extrema, a local maximum
at $S_0=0$ where $h_k(0)=2^{k-1}>0$ and a local minimum at $S_0=r_{\min}(k)=\frac{2}{3k(2\alpha_0)^{k-1}}$. Since
$h_k(r_{\min}(k))<0$ if and only if $k\geq 5$, $h_k$ has two positive roots when $k\geq5$, namely $r_{\text{left}}(k)$
and $r_{\text{right}}(k)$. Therefore, $h_k(S_0)\geq0$ is equivalent to $S_0\in [0,r_{\text{left}}(k)]\cup[r_{\text{right}}(k),\infty)$ when $k\geq5$.

We will obtain a contradiction by showing that $S_0\in(r_{\text{left}}(k),r_{\text{right}}(k))$ for some $k\geq 5$. This is necessarily true since $\cup_{k\geq 5}(r_{\text{left}}(k),r_{\text{right}}(k))=(r_{\text{left}}(5),\infty)\supseteq (r_{\min}(5),\infty)\supseteq[9.62,\infty)\ni S_0$.

It only remains to be proven that $\cup_{k\geq 5}(r_{\text{left}}(k),r_{\text{right}}(k))=(r_{\text{left}}(5),\infty)$. This follows from $r_{\text{right}}(k)>r_{\min}(k)=\frac{2}{3k(2\alpha_0)^{k-1}}\longrightarrow\infty$ as $k\to\infty$, and the fact that $r_{\text{left}}(k+1)<r_{\text{right}}(k)$ for all $k\geq 5$. Indeed, since $h_{k+1}(r_{\min}(k))<0$ and $r_{\min}(k)<r_{\text{right}}(k)$, then $r_{\text{left}}(k+1)<r_{\min}(k)<r_{\text{right}}(k)$.
\end{proof}

\section{Conclusions}

We provided in Theorem~\ref{thm-starterKeplerS} a very simple starter $\widetilde{S}(g,L)$ for the hyperbolic Kepler's
equation which uses a single addition and multiplication when $L>1-\frac56g$ and a few arithmetic operations, a cubic
and square root otherwise. Our starter can be implemented efficiently on modern computers and has guaranteed quadratic convergence speed from the first iteration, so the relative error becomes negligible after just a few iterations for any value of $g\in(0,1)$ and $L\in[0,\infty)$.

If an even simpler starter is necessary, we showed in Theorem~\ref{thm-const-cubre} that it is possible to produce a piecewise constant starter in any bounded region of $(0,1)\times[0,\infty)$, excluding a neighborhood of the corner $g=1,L=0$. This starter also has quadratic convergence rate and its evaluation requires no operations, so it can be efficiently implemented as a look-up table.

\section*{Acknowledgements}
The first author is partially supported by the MINECO grant ESP2013-44217-R, the second author by the MINECO grant MTM2011-22621 and the FQM-327 group (Junta de Andaluc\'ia, Spain).

\end{document}